\documentclass[10pt,twocolumn,letterpaper]{article}

\usepackage{iccv}
\usepackage{times}
\usepackage{epsfig}
\usepackage{graphicx}
\usepackage{amsmath}
\usepackage{amssymb}

\usepackage[absolute]{textpos}
\usepackage{appendix}
\usepackage[ruled,vlined]{algorithm2e}
\SetArgSty{textnormal}
\usepackage{amsthm}
\usepackage{subcaption}
\usepackage{dcolumn}
\usepackage{booktabs}

\usepackage[pagebackref=true,breaklinks=true,letterpaper=true,colorlinks]{hyperref}
\usepackage{csquotes}

\newcommand{\setcomp}[2]{\left\{ #1 \,\middle|\, #2 \right\}}
\newcommand{\norm}[1]{\left\lVert#1\right\rVert}
\newcommand{\abs}[1]{\left\lvert#1\right\rvert}

\DeclareMathOperator*{\argmin}{arg\,min\:}

\newcommand{\eps}{\varepsilon}
\newtheorem{theorem}{Theorem}

\newtheorem{lemma}[theorem]{Lemma}
\newtheorem{corollary}[theorem]{Corollary}
\theoremstyle{definition}
\newtheorem{definition}{Definition}
\newcolumntype{d}{D{?}{\,\pm\,}{-1}}

\newenvironment{subproof}[1][\proofname]{%
\begin{proof}[#1]%
}{%
\end{proof}%
}

\graphicspath{ {./fig/} }

\iccvfinalcopy 

\def\iccvPaperID{6984} 

\ificcvfinal\fi

\begin{document}

\title{Extensions of Karger's Algorithm: Why They Fail in Theory and How They Are Useful in Practice}

\author{
Erik Jenner \qquad Enrique Fita Sanmartín \qquad Fred A. Hamprecht \\
  Heidelberg Collaboratory for Image Processing\\
  University of Heidelberg, Germany\\
  {\tt\small erik@ejenner.com, \{enrique.fita.sanmartin,fred.hamprecht\}@iwr.uni-heidelberg.de}
}

\maketitle

\begin{textblock*}{\textwidth}(2cm,26.5cm)
  \footnotesize
© 2021 IEEE.  Personal use of this material is permitted.  Permission from IEEE must be obtained for all other uses, in any current or future media, including reprinting/republishing this material for advertising or promotional purposes, creating new collective works, for resale or redistribution to servers or lists, or reuse of any copyrighted component of this work in other works.
\end{textblock*}

\ificcvfinal\thispagestyle{empty}\fi

\begin{abstract}
  The minimum graph cut and minimum \(s\)-\(t\)-cut problems are important
  primitives in the modeling of combinatorial problems in computer science,
  including in computer vision and machine learning. Some of the most efficient
  algorithms for finding global minimum cuts are randomized algorithms based on
  Karger’s groundbreaking contraction algorithm. Here, we study whether Karger’s
  algorithm can be successfully generalized to other cut problems. We first
  prove that a wide class of natural generalizations of Karger’s algorithm
  cannot efficiently solve the \(s\)-\(t\)-mincut or the normalized cut problem
  to optimality. However, we then present a simple new algorithm for seeded
  segmentation / graph-based semi-supervised learning that is closely based on
  Karger's original algorithm, showing that for these problems, extensions of
  Karger's algorithm can be useful. The new algorithm has linear asymptotic
  runtime and yields a potential that can be interpreted as the posterior
  probability of a sample belonging to a given seed / class. We clarify its
  relation to the random walker algorithm / harmonic energy minimization in
  terms of distributions over spanning forests. On classical problems from
  seeded image segmentation and graph-based semi-supervised learning on image
  data, the method performs at least as well as the random walker / harmonic
  energy minimization / Gaussian processes.
\end{abstract}
\section{Introduction}
Minimum graph cuts have been applied to machine learning problems for a long
time. They have been used in natural language
processing~\cite{pangSentimentalEducationSentiment2004} and especially in
computer vision, for example in
segmentation~\cite{wuOptimalGraphTheoretic1993,rotherGrabCutInteractiveForeground2004,boykov2001},
restoration~\cite{greigExactMaximumPosteriori1989}, and energy minimization more
generally~\cite{kolmogorovWhatEnergyFunctions2004}. Nowadays, they still form an
important part of many deep learning
pipelines, for example for
segmentation~\cite{xu2016,lu2017,mukherjee2017,li2019,liu2019},
image classification~\cite{nardelli2018},
and recently also neural style transfer~\cite{zhang2019a}.

For finding global minimum cuts (defined together with all other terminology in
Section~\ref{sec:background}), Karger's contraction
algorithm~\cite{kargerGlobalMincutsRNC1993,kargerNewApproachMinimum1996}
started a wave of randomized algorithms solving this problem efficiently~\cite{
  kargerMinimumCutsNearLinear1998,gawrychowski2019,ghaffariFasterAlgorithmsEdge2019,lovett2020}.

Thanks to these randomized algorithms,  global mincuts can, somewhat surprisingly, be found
more efficiently than \(s\)-\(t\)-mincuts. An interesting question is therefore
to what extent randomized algorithms can be applied to other graph cut problems,
and in particular whether Karger's algorithm can be fruitfully extended.
An especially important cut problem are \(s\)-\(t\)-mincuts. While
\emph{approximating} them is possible in nearly linear time in
the number of edges~\cite{kelner2013,sherman2013} and has also been studied
using randomized algorithms based on graph sparsification~\cite{benczurApproximatingStMinimum1996},
it is to our knowledge still an open question whether Karger's
algorithm can be modified to efficiently find \(s\)-\(t\)-mincuts.

In Section~\ref{sec:impossibility}, we give a definitive answer to this question by proving that a large class
of extensions of Karger's contraction algorithm can in general not exactly solve the \(s\)-\(t\)-mincut problem
efficiently. Our result also applies to the normalized cut
problem~\cite{jianboshiNormalizedCutsImage2000}, which, like the
\(s\)-\(t\)-mincut, plays an important role in image segmentation.

However, extensions of Karger's algorithm can still be useful if applied in the
right way. In Section~\ref{sec:seeded_segmentation}, we show how a
straightforward extension of Karger's algorithm can be used successfully
for seeded segmentation / semi-supervised learning tasks. We interpret this
extension as a forest sampling method and observe its similarities to the random
walker algorithm~\cite{gradyRandomWalksImage2006} for seeded graph segmentation.
In semi-supervised learning, the same algorithm is known as harmonic energy
minimization~\cite{zhu2003} or Gaussian Processes, so the same observations
apply.

The main contribution of this paper is purely conceptual. Still, in Section~\ref{sec:experiments}
we show in two classical experiments that the proposed algorithm compares well
against the random walker / harmonic energy minimization, perhaps the most
influential algorithm in seeded segmentation / semi-supervised learning
to date.
Since our method has an asymptotic time complexity of only \(\mathcal{O}(m)\) on a graph with \(m\) edges,
it can be seen as an efficient alternative to the random walker algorithm / harmonic energy minimization,
while also giving a probabilistic output.

\paragraph{Related work} The most closely related work is the \emph{typical cut}
algorithm~\cite{gdalyahuSelforganizationVisionStochastic2001},
which uses an ensemble of cuts generated by Karger's algorithm for clustering without seeds. In contrast, the method described
here uses cuts generated by a slight variation of Karger's algorithm to solve \emph{seeded} segmentation problems.
In this setting, there is a very natural way to get a segmentation from the ensemble of cuts, as well as a natural stopping
point for the contraction, which for the typical cut is a free parameter.

\section{Background}\label{sec:background}
All graphs considered in this paper are undirected and connected and have non-negative edge weights.
We write such a graph as a tuple \(G = (V, E, w)\) of a set of vertices \(V\), an edge set \(E\)
and a weight function \(w: E \to \mathbb{R}_{\geq 0}\). We denote the number of
vertices by \(n := \abs{V}\) and the number of edges by \(m := \abs{E}\).
We also write \(w_e\) for the weight \(w(e)\) of an edge \(e \in E\) and \(w_{uv}\) for the weight
of the edge between vertices \(u, v \in V\). If no edge is present, \(w_{uv}\) is defined
as zero.
\(w(A, B) := \sum_{a \in A, b \in B} w_{ab}\) is the sum of edge weights connecting two
subsets \(A, B \subset V\).

A \emph{graph cut} is a partition of the vertices \(V\) of a graph into
two disjoint non-empty subsets \(A\) and \(B\) such that \(V = A \cup B\).
The \emph{cut set} of such a cut is the set of all edges with one endpoint
in \(A\) and one in \(B\). The sum \(w(A, B)\) of the weights of all edges in the
cut set is called the \emph{weight} or \emph{cost} of the graph cut.

We will describe three different cut problems here: the global minimum cut,
the \(s\)-\(t\)-minimum cut and the normalized cut.
\paragraph{A \emph{(global) minimum cut}} of a graph -- or \emph{mincut} for short -- is a cut with minimal cost.
In other words, the minimum cut problem is given by
\begin{equation}
  \argmin_{\text{partitions } (A, B) \text{ of } V} w(A, B)\,.
\end{equation}

\paragraph{An \emph{\(s\)-\(t\)-cut}} of a graph \(G\) is a graph cut that separates two given
vertices \(s \neq t \in V\).
In the \emph{\(s\)-\(t\)-mincut} problem, the goal is to find an \(s\)-\(t\)-cut
with minimal cost, i.e.
\begin{equation}
  \argmin_{\text{partitions } (S, T)} w(S, T) \quad \text{such that } s \in S, t \in T\,.
\end{equation}
For convenience, we define an \emph{\(s\)-\(t\)-graph} as a tuple \((G, s, t)\)
of a graph and two vertices \(s \neq t \in V.\)

We also mention here the notion of \emph{\(\alpha\)-minimal cuts}. A global cut \((A, B)\) is \(\alpha\)-minimal
if its cost is within a factor \(\alpha\) of the global minimum cut,
\begin{equation}
  w(A, B) \leq \alpha \min_{\substack{\text{partitions }\\ (A', B')}} w(A', B')\,,
\end{equation}
where \(\alpha\) is some positive real number \(> 1\). The same concept can of course be applied to
define an \(\alpha\)-minimal \(s\)-\(t\)-cut as an \(s\)-\(t\)-cut that has a cost
within a factor \(\alpha\) of the \(s\)-\(t\)-mincut.

\paragraph{The \emph{normalized cut}}~\cite{jianboshiNormalizedCutsImage2000}
generates more balanced cuts than the minimum cut objective, which makes it
particularly well suited for image segmentation.
It minimizes
\begin{equation}
\operatorname{ncut}(A, B) := \frac{w(A, B)}{w(A, V)} + \frac{w(A, B)}{w(B, V)}
\end{equation}
over the partitions \((A, B)\) of the graph.
Note that since \(w(A, V) \overset{\text{(def)}}{=} \sum_{a \in A, v \in V} w_{av}\), this term counts the internal
weights of \(A\) twice.
Solving the normalized cut problem exactly is NP-complete, but the solution
can be approximated with a spectral method~\cite{jianboshiNormalizedCutsImage2000}.

\subsection{Karger's contraction algorithm}
Karger's algorithm is a Monte Carlo algorithm for finding global minimum graph cuts,
meaning that it has a fixed runtime but is not guaranteed to find the best cut.
It is based on \emph{contractions} of edges in a graph. Given a graph \(G = (V, E, w)\)
and two vertices \(v_1, v_2 \in V\), the contracted graph \(G/\{v_1, v_2\}\) is obtained as follows:
\begin{enumerate}
\item \(v_1\) and \(v_2\) with all their edges are removed and a new vertex \(u\) is added.
\item For each edge \(\{v_i, x\} \in E\) with \(x \notin \{v_1, v_2\}\), a new edge \(\{u, x\}\) with the same weight
is added, for \(i = 1, 2\).
\item If \(u\) now has several edges to the same vertex, they are merged into one by adding their weights.
\end{enumerate}

Karger's algorithm simply repeatedly chooses an edge at random and contracts it
until only two vertices remain. The remaining edges then define a cut set. Each
edge is chosen for contraction with probability proportional to its weight.
The precise algorithm is described in algorithm~\ref{alg:karger}.

\begin{algorithm}
\SetKwInOut{Input}{Input}
\SetKwInOut{Output}{Output}
\Input{graph \(G\)}
\Output{contracted graph with 2 vertices}
\While{\(G\) has more than 2 vertices}{
    choose an edge \(\{u, v\}\) with probability proportional to its weight\;
    \(G \gets G/\{u, v\}\)\;
}
\Return{\(G\)}\;
\caption{Karger's contraction algorithm}\label{alg:karger}
\end{algorithm}
Of course, this algorithm does not always produce a minimum cut. To increase the
success probability, the algorithm is run several times and the best
cut is returned. This can be sped up by sharing computations between
runs~\cite{kargerN2AlgorithmMinimum1993,kargerNewApproachMinimum1996}
but doing so does not affect any of the arguments in this paper, so we will ignore it.

The reason why Karger's algorithm is useful for finding minimum cuts is the following theorem, which says that
-- compared to the success probability of \(2^{-n + 1}\) that uniform sampling of cuts would give --
Karger's algorithm finds a minimum cut with relatively high probability
on a single run. This means that a polynomial number of runs is enough to
find a minimum cut with high probability.
\begin{theorem}[\cite{kargerGlobalMincutsRNC1993}]\label{thm:karger_probability_bound}
The probability of finding any given mincut with Karger's algorithm
is at least \({n \choose 2}^{-1}\).
\end{theorem}
The key idea of the proof is that the cost of a global minimum cut is only a small fraction
of the sum of all edge weights because it is always possible to cut out only the vertex
with the lowest degree, which gives an upper bound of \(\frac{2}{n} \sum_{e \in E}w_e\)
for the cost of any minimum cut. So because the contraction probabilities are
proportional to the edge weights, it is -- at least initially -- unlikely that an edge
which is part of a minimum cut set will be contracted.

We will show that an analog of Theorem~\ref{thm:karger_probability_bound}
does not exist for \(s\)-\(t\)-mincuts or normalized cuts, even for a wide class of extensions
of Karger's algorithm. These algorithms would need to be run an exponential
number of times in some cases to obtain a high success probability.

\subsection{Random walker / harmonic energy minimization}
Both global minimum cuts and the normalized cut problem are unsupervised approaches
to clustering: they take only a graph as input, without any annotations.

In contrast, in the seeded segmentation / semi-supervised learning problem, labels are given for
some vertices, the \emph{seeds}. The goal is to assign fitting labels to the
remaining vertices.
This problem can occur in different contexts: in image segmentation,
each vertex corresponds to a pixel, while in graph-based semi-supervised learning,
each vertex represents one sample and the seeds are the labeled samples.

One method for solving the seeded segmentation problem is the
\emph{random walker algorithm}~\cite{gradyRandomWalksImage2006},
also known as \emph{harmonic energy minimization}~\cite{zhu2003} or Gaussian Processes
in graph-based semi-supervised learning.
To choose a label for some vertex \(v\), it imagines a random walker on the graph starting
on \(v\). This random walker chooses an edge to traverse with probability proportional
to the edge weight at each step. It stops once it reaches one of the seeds.
We write \(p_{\text{rw}}(v \sim l)\) for the probability that the random
walker reaches a seed with label \(l\) when starting from \(v\), which we also call
the \emph{random walker potential}. Each vertex is assigned
to the label for which this probability is highest.

Actually simulating such a random walker for each vertex would be intractable.
But the probabilities \(p_{\text{rw}}(v \sim l)\) can be calculated by solving a linear system containing the Laplacian
of the graph~\cite{gradyRandomWalksImage2006,zhu2003}. This means finding an
approximate solution is possible in nearly-linear time in the number of edges
using fast Laplacian solvers~\cite{spielman2004,koutis2010,koutis2011,cohen2014}.

The random walker can also be interpreted as a forest sampling method.
We write \(\mathcal{F}_s\) for the set of spanning forests
of the graph where each tree spans all seeds of a given category, and the non-intersecting trees together span the graph.
Any such forest defines a label for each vertex \(v\).
We can define a Gibbs distribution over these forests by
\begin{equation}\label{eq:gibbs_forest}
p(f) = \frac{1}{Z} \prod_{e \in f} w_e = \frac{1}{Z} w(f)
\end{equation}
for a forest \(f \in \mathcal{F}_s\) with weight \(w(f) := \prod_{e \in f} w_e\).
The partition function is given by \(Z := \sum_{f \in \mathcal{F}_s} w(f)\).
It can then be shown~\cite{grady_discrete_2010,fitasanmartinProbabilisticWatershedSampling2019}
that the probability with which a forest sampled from this distribution
assigns a vertex \(v\) to the label \(l\) is precisely the random walker probability \(p_{\text{rw}}(v \sim l)\).

\section{Impossibility results}\label{sec:impossibility}
In this section, we present a framework that greatly generalizes Karger's
algorithm to what we call \emph{general contraction algorithms}. We then show
that algorithms from two natural subsets of this class of algorithms cannot be
used to efficiently find \(s\)-\(t\)-mincuts or normalized cuts.

General contraction algorithms are described formally in
algorithm~\ref{alg:contraction}. Like Karger's algorithm, they sample and
contract edges until two vertices remain. But the contraction probabilities may
now depend on arbitrary graph properties, rather than being proportional to the
edge weights.

\begin{algorithm}
\SetKwInOut{Input}{Input}
\SetKwInOut{Output}{Output}
\Input{graph \(G\), optionally with seeds \(s\) and \(t\) (depending on the algorithm)}
\Output{contracted graph with 2 vertices}
\While{\(G\) has more than 2 vertices}{
    \(A \gets\) weighted adjacency matrix of \(G\)\;
    choose an edge \(e\) with probability proportional to \(\mathcal{W}(e; A, s, t)\)\;
    \(G \gets G/e\)\;
}
\Return{\(G\)}\;
\caption{The general contraction algorithm.
When \(s\) / \(t\) is contracted with another node, the new node becomes
the new \(s\) / \(t\). The score function \(\mathcal{W}\)
distinguishes different contraction algorithms.}\label{alg:contraction}
\end{algorithm}

Any contraction algorithm is fully defined by specifying the score
\(\mathcal{W}(e; A, s, t)\) it assigns to an edge \(e\) in a graph with weighted adjacency
matrix \(A\) and seed indices \(s\) and \(t\) (where \(e\) is a two-set \(\{i, j\}\) of vertices).
The weighted adjacency matrix contains the edge weights, i.e.\ \(A_{ij} = w_{ij}\) is the weight
between vertices \(i\) and \(j\).

Karger's algorithm is clearly the special case with \(\mathcal{W}(e) = w_e\),
or more explicitly, \(\mathcal{W}(\{i, j\}; A, s, t) = A_{ij} = A_{ji}\).
As another example, we can define the following modification of Karger's
algorithm:
\begin{equation}
  \mathcal{W}(\{i, j\}; A, s, t) =
  \begin{cases}
    0, &\quad \{i, j\} = \{s, t\}\\
    A_{ij}, &\quad \text{otherwise}
  \end{cases}\,.
\end{equation}
This contraction algorithm, which we call the \emph{\(s\)-\(t\)-contraction algorithm},
never contracts edges connecting \(s\) and \(t\) and therefore always samples
an \(s\)-\(t\)-cut. It is relatively easy to show that this particular extension
of Karger's algorithm finds \(s\)-\(t\)-mincuts with only very low probability
on some graphs (we will shortly give a simple proof). However, the
framework of general contraction algorithms also includes choices that always find
\(s\)-\(t\)-mincuts, such as
\begin{equation}
  \mathcal{W}(e; A, s, t) := \begin{cases}
    0, & e \in C \\
    1, &\text{otherwise}
  \end{cases}
\end{equation}
for the cut set \(C\) of some \(s\)-\(t\)-mincut.
Of course this specific
method is impractical because calculating the weights requires already knowing
an \(s\)-\(t\)-mincut, but it demonstrates that contraction algorithms can in
principle find \(s\)-\(t\)-mincuts with high probability. What is a priori
unclear is whether any \emph{practical} contraction algorithm can do so.

To answer this question, we introduce two natural and very general classes of contraction
algorithms, for which we can formally prove impossibility results:
\emph{continuous} contraction algorithms and \emph{local} ones.

\begin{definition}
  A \emph{continuous contraction algorithm} is a general contraction algorithm (see algorithm~\ref{alg:contraction})
  whose score \(\mathcal{W}\) is a continuous function of the adjacency matrix \(A\).
\end{definition}
Intuitively, this means that slight changes in the weights of a graph lead to only slight changes in the contraction
probabilities for continuous contraction algorithms. Since the same results hold for finding \(s\)-\(t\)-mincuts and normalized cuts, we state
them together:
\begin{theorem}\label{thm:continuous_impossibility}
  For any continuous contraction algorithm, there is a family of \(s\)-\(t\)-graphs (graphs) on which
  it finds an \(s\)-\(t\)-mincut (normalized cut) with only exponentially low probability in the number of vertices.
\end{theorem}
The full proof of Theorem~\ref{thm:continuous_impossibility} and all other results can be found
in the supplementary material.
The idea of the proof is to take a graph in which there are exponentially many different
\(s\)-\(t\)-mincuts (normalized cuts). Then there must be at least one such cut that is chosen with
exponentially low probability. If the weights are perturbed slightly to make this cut
the unique \(s\)-\(t\)-mincut (normalized cut), the probability of sampling it will remain low.
The reason that this proof does not apply to global minimum cuts is that there are at most
\({n \choose 2}\) global minimum cuts in any graph, as Theorem~\ref{thm:karger_probability_bound}
implies.

We now come to our second impossibility result, that for \enquote{local} contraction
algorithms.
\begin{definition}
  The \emph{neighborhood} \(N(e)\) of an edge \(e = \{u, v\} \in E\) is the subgraph
  of \(G\) induced by the neighbors of \(u\) and \(v\). It consists of the vertex set \(V_{N(e)} := \setcomp{x \in V}{x \text{ is neighbor of } u \text{ or } v}\)
  and of all edges from \(E\) connecting pairs of vertices from that set.
\end{definition}
We treat two neighborhoods \(N(e_1), N(e_2)\) as the same if there is a graph isomorphism \(f: V_{N(e_1)} \to V_{N(e_2)}\)
that also preserves \(s\) and \(t\) if applicable, i.e.\ \(f(s) = s, f(t) = t\).
\begin{definition}
A general contraction algorithm is \emph{local} if the score \(\mathcal{W}(e; A, s, t)\) can be written
as a function \(\mathcal{W}(N(e), G)\).
\end{definition}

Informally speaking, a local contraction algorithm assigns scores based only
on local properties of the edges and on global properties of the entire graph.
It does not have access to properties of the individual edges that depend on their placement
in the graph.

For this class of algorithms, we can prove a similar result as for continuous
contraction algorithms:
\begin{theorem}
  There is a family of \(s\)-\(t\)-graphs (graphs)
  on which any local contraction algorithm finds an \(s\)-\(t\)-mincut (normalized cut)
  with only exponentially low probability.
\label{thm:local_impossibility}
\end{theorem}
\begin{figure}[htbp]
\centering
\includegraphics[width=.5\linewidth]{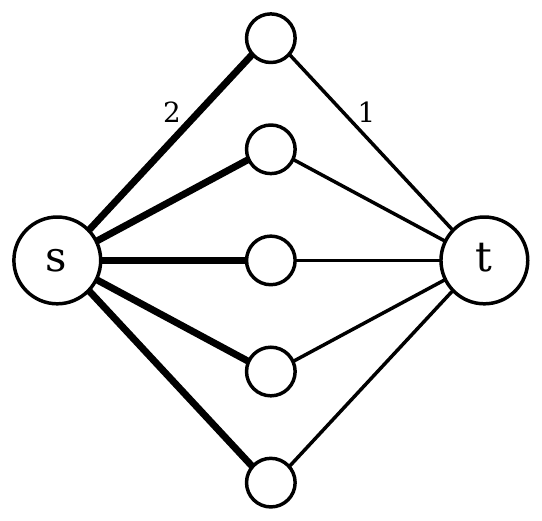}
\caption{\label{fig:simple_counterexample}Graph where the
  \(s\)-\(t\)-contraction algorithm performs badly. The thick edges have a higher weight.}
\end{figure}
To illustrate the idea of the proof, consider the graph shown in Fig.~\ref{fig:simple_counterexample}.
This graph can be used to prove Theorem~\ref{thm:local_impossibility} for the \(s\)-\(t\)-contraction algorithm
(instead of for local contraction algorithms in general) as follows:
If we choose a weight of 1 for the thin edges and 2 for the
thicker edges, then there is a unique \(s\)-\(t\)-mincut. To find this cut, only
thick edges may be contracted during all \(n - 2\) contractions.
But the probability of choosing a thick edge for contraction is always only \(\frac{2}{3}\).
So the overall success probability is
\begin{equation}
  p_{\text{success}} = \left(\frac{2}{3}\right)^{n - 2}
\end{equation}
If we scale up the graph in Fig.~\ref{fig:simple_counterexample}, this success
probability diminishes exponentially in the number of vertices.

The general proof for all local contraction algorithms (see supplementary material)
uses the same idea of a graph with many parallel paths between \(s\) and \(t\), each of which has
to be contracted correctly independently. Those paths are more complex than in
Fig.~\ref{fig:simple_counterexample} and chosen such that it is impossible to
decide whether an edge belongs to the \(s\)-\(t\)-mincut based only on local properties.

The same proof idea implies that local contraction algorithms
cannot even approximate the \(s\)-\(t\)-mincut beyond some threshold with high probability:
\begin{corollary}\label{thm:local_approximability}
  The probability of finding an \(\alpha\)-minimal \(s\)-\(t\)-cut of the graphs from
  Theorem~\ref{thm:local_impossibility} is exponentially low
  for all local contraction algorithms if \(\alpha < 2\).
\end{corollary}
The threshold of 2 does not carry a deep meaning. It just comes from
the particular graph we used for the proof and the statement may hold for a larger threshold.
Note that this result is only stated for \(s\)-\(t\)-mincuts, not for normalized cuts.
Since normalized cut costs are always in \([0, 2]\), the proof does not transfer
as it did for the other theorems.

\section{Seeded contraction algorithm}\label{sec:seeded_segmentation}
The results from the previous section show that sampling cuts using local or
continuous contraction algorithms and then taking the smallest cut out of the
population sampled this way does not necessarily give a minimum cut. However,
this population can be used in other ways. In this section, we describe a new
method for seeded graph segmentation that can be interpreted as computing the
\emph{mean} of the sampled cuts, rather than the single smallest cut.
We also describe theoretical similarities between our method and the
random walker algorithm / harmonic energy minimization. In the next section, we will compare these two methods
empirically.

To make the new method widely applicable, we first generalize the
\(s\)-\(t\)-contraction algorithm from the previous section to more than two
labels and multiple seeds per label.
The problem setup consists of a weighted graph \(G = (V, E, w)\) and a surjective seed function
\(s: V \to \{0, \ldots, k\}\) where \(k\) is the number of labels and 0 is assigned to unlabeled nodes.

A given cut \(V = V_1 \cup \ldots \cup V_k\) into disjoint vertex subsets \emph{respects}
the seeds \(s\) if \(s(v) = l \implies v \in V_l\) for all \(l \in \{1, \ldots k\}\) and \(v \in V\). Such a cut defines a labeling
of the entire graph, by assigning label \(l\) to vertex \(v\) if \(v \in V_l\).

In the special case of \(k = 2\) and only one seed per class, these cuts
are simply \(s\)-\(t\)-cuts, which can be sampled with the \(s\)-\(t\)-contraction algorithm
from the previous section. The \emph{seeded contraction algorithm}
(algorithm~\ref{alg:seeded_contraction}), generalizes this and produces cuts that respect the input seeds
for arbitrary numbers of classes and seeds per class.

\begin{algorithm}
\SetKwInOut{Input}{Input}
\SetKwInOut{Output}{Output}
\Input{graph \(G = (V, E, w)\), labels \(s: V \to \{0, \ldots, k\}\)}
\Output{contracted graph with \(k\) vertices}
Contract all edges between nodes with the same label\;
Remove edges between nodes with different labels\;
\While{\(G\) has more than \(k\) vertices}{
    choose an edge \(\{v_1, v_2\}\) with probability proportional to its weight\;
    \(G \gets G/\{v_1, v_2\}\)\;
    \uIf{\(v_1\) or \(v_2\) has a label}{
        assign the new node created by merging \(v_1\) and \(v_2\) that label\;
    }
    Remove edges between nodes with different labels\;
}
\Return{\(G\)}\;
\caption{Seeded contraction algorithm with \(k\) different labels. A label of 0 means \enquote{no label}.}\label{alg:seeded_contraction}
\end{algorithm}

For \(l \in \{1, \ldots k\}\), we define \(p_{\text{contr}}(v \sim l)\) as the probability
that the seeded contraction algorithm produces a cut which assigns label \(l\) to
the vertex \(v\). Because this algorithm is a very natural extension of Karger's algorithm
to seeded segmentation, we will also refer to this distribution as the \enquote{Karger potential}.

The seeded contraction algorithm can be run multiple times to approximately find
the probabilities \(p_{\text{contr}}(v \sim l)\) for each vertex \(v\) and label
\(l\). If a hard assignment is required, each vertex can then be assigned to the
label for which this probability is highest.

To compare the Karger potential to the random walker potential, we reinterpret the seeded contraction
algorithm as a forest sampling method. During a single run of the contraction algorithm,
\(n - k\) edges are selected for contraction. These edges form a spanning \(k\)-forest of
the graph, where each component of the forest is one of the subsets \(V_l\) of the cut.
So our method defines a probability distribution over the set \(\mathcal{F}_s\) of \(k\)-forests that separate the
seeds with different labels. \(p_{\text{contr}}(v \sim l)\) is the probability that a forest
sampled from this distribution connects \(v\) to the seeds with label \(l\).

This is reminiscent of the random walker distribution \(p_{\text{rw}}\) which
can be interpreted as the probability that a forest sampled from a Gibbs distribution
connects \(v\) to the seeds with label \(l\). The only difference between the two methods
is the distribution over forests they use.

To understand the effects of this difference, we will derive an expression for the probability
that the seeded contraction algorithm samples a given forest.

For a subset \(\hat{E} \subset E\) of edges, we define
\begin{equation}
  \begin{split}
    \mathcal{C}(\hat{E}) := \hat{E} \cup
    \{e \in E | e \cup \hat{E} \text{ has cycles or contains}\\
      \text{a path between seeds with different labels}\}\,.
  \end{split}
\end{equation}
\(\mathcal{C}(\hat{E})\) is precisely the set of edges that
has been removed after the edges from \(\hat{E}\) have been contracted because
each edge that forms a cycle with those in \(\hat{E}\) has become a
self-loop. We write \(c(\hat{E}) := \sum_{e \in \hat{E}} w_e\) for the sum
of weights of a set of edges. Then the total weights of edges remaining after contracting the edges from \(\hat{E}\)
will be \(c\left(E \setminus \mathcal{C}\left(\hat{E}\right)\right)\).

Therefore, the probability of contracting edges \(e_1, \ldots, e_{n - 2}\) in
that order is
\begin{equation}
p(e_1, \ldots, e_{n - 2}) = \prod_{i = 1}^{n - 2} \frac{w(e_i)}{c\left(E \setminus \mathcal{C}(\{e_1, \ldots, e_{i - 1}\})\right)}\,.
\end{equation}
Note the \(i - 1\) in the denominator; the term describes the probability at the \(i\)th contraction
step, at which point only \(e_1, \ldots, e_{i - 1}\) have been contracted.

For the sampled forest \(f\), it does not matter in which order its constituent
edges \(e_1, \ldots, e_{n - 2}\) are contracted, so the total probability is
\begin{equation}\label{eq:karger_forest_distribution}
\begin{split}
p(f) &= \sum_{\sigma \in S_{n - 2}} \prod_{i = 1}^{n - 2}
\frac{w(e_{\sigma(i)})}{c\left(E \setminus\mathcal{C}(\{e_{\sigma(1)}, \ldots, e_{\sigma(i - 1)}\})\right)}\\
&= w(f) \sum_{\sigma \in S_{n - 2}} \prod_{i = 1}^{n - 2}
\frac{1}{c\left(E \setminus \mathcal{C}(\{e_{\sigma(1)}, \ldots, e_{\sigma(i - 1)}\})\right)}\,.
\end{split}
\end{equation}

We can compare this distribution to the Gibbs distribution over 2-forests
that the random walker algorithm samples from,
\begin{equation}\label{eq:gibbs_repeated}
p(f) = \frac{1}{Z} \prod_{e \in f} w_e = \frac{1}{Z} w(f)\,,
\end{equation}
where \(Z = \sum_{f \in \mathcal{F}_s} w(f)\). Both distributions contain the
term \(w(f)\) but where the Gibbs distribution
has a partition function \(Z\) that is independent of the forest \(f\),
the distribution of the contraction algorithm has the sum over permutations term with an
additional dependency on \(f\).

Note that \(w(e_1), \ldots, w(e_{n - 2})\) all contribute to the cost of
\(\mathcal{C}(\{e_1, \ldots, e_{n - 2}\})\). So a 2-forest with large edge weights has a high
probability not just because of the term \(w(f)\) but also because of the second term
in eq.~\eqref{eq:karger_forest_distribution}. This means that compared to the Gibbs distribution
from eq.~\eqref{eq:gibbs_repeated}, we expect the contraction distribution to favor heavy forests
more strongly.

Therefore, the Karger potential should be \enquote{more confident} than the random walker
potential -- both will typically be highest for the same label \(l\), but \(p_{\text{contr}}(v \sim l)\) will be
higher than \(p_{\text{rw}}(v \sim l)\) for that label.

There is a second effect which is of a topological nature: the cost of \(\mathcal{C}(\{e_1, \ldots, e_{n - 2}\})\)
will tend to be large if \(\mathcal{C}(\{e_1, \ldots, e_{n - 2}\})\) contains many edges.
Since \(e_1, \ldots, e_{n - 2}\) is a 2-forest, the only edges not in that set are precisely
the edges in the cut set that the 2-forest induces. So this is again a reason to think
that the Karger distribution assigns more extreme probabilities than the Gibbs distribution
-- a large weight of the forest is equivalent to a small weight of the induced cut.

There is a big difference in how the Karger and random walker potential can be calculated
in practice. As mentioned, the random walker potential can be calculated exactly by solving
a system of linear equations.
In contrast, calculating the Karger potential exactly appears to be infeasible for all but the
smallest graphs. However, the seeded contraction algorithm can be used to
efficiently sample from the distribution \(p_{\text{contr}}\) and by running it
multiple times, this distribution can be approximated.

To achieve a fixed precision in the approximation, the seeded contraction algorithm
needs to be run only a constant number of times, independent of the size of the graph.
Our segmentation method therefore has a runtime complexity of only \(\mathcal{O}(m)\),
where \(m\) is the number of edges of the graph (details on how to implement the seeded
contraction algorithm in \(\mathcal{O}(m)\) time can be found in the supplementary material).

\section{Experiments}\label{sec:experiments}
We compare the new segmentation method from the previous section to the random walker
on an image segmentation and a semi-supervised learning task.
To keep the focus on the methods under comparison,
rather than the rest of the pipelines, we chose two classical tasks and
well-known, relatively simple pipelines for computing the edge weights.
All of our code can be found at
\url{https://github.com/ejnnr/karger_extensions}. A few additional details, such
as empirical runtimes, are part of the supplementary material.

\paragraph{Seeded segmentation} We use the Grabcut~\cite{rotherGrabCutInteractiveForeground2004}
images with sparse labels from~\cite{gulshanGeodesicStarConvexity2010}.
To create graphs from images, we used the usual
4-connected topology, meaning that each pixel is connected by an edge to
its four neighbors (or fewer at the border).

We obtained edge weights with holistically-nested edge detection~\cite{xie2015}
using a PyTorch implementation~\cite{niklaus2020}. This yields an intensity \(g_i \in [0, 1]\)
(after dividing by the maximum intensity) for each pixel,
where higher values correspond to edges recognized by the network.
For the edge weights, we then used
\begin{equation}
  w_{ij} = \exp\left(-\beta (g_i + g_j)^2\right)\,,
\end{equation}
where \(\beta\) is a free parameter.

Figure~\ref{fig:potential} shows the effect of \(\beta\) on one of the Grabcut
images.
Note that for intermediate values of \(\beta\), e.g.\ \(\beta = 5\), we can see the
higher \enquote{confidence} of the Karger potential compared to the random walker
potential, as hypothesized in Section~\ref{sec:seeded_segmentation}.

\begin{figure*}[htb]
  \centering
  \includegraphics[width=\linewidth]{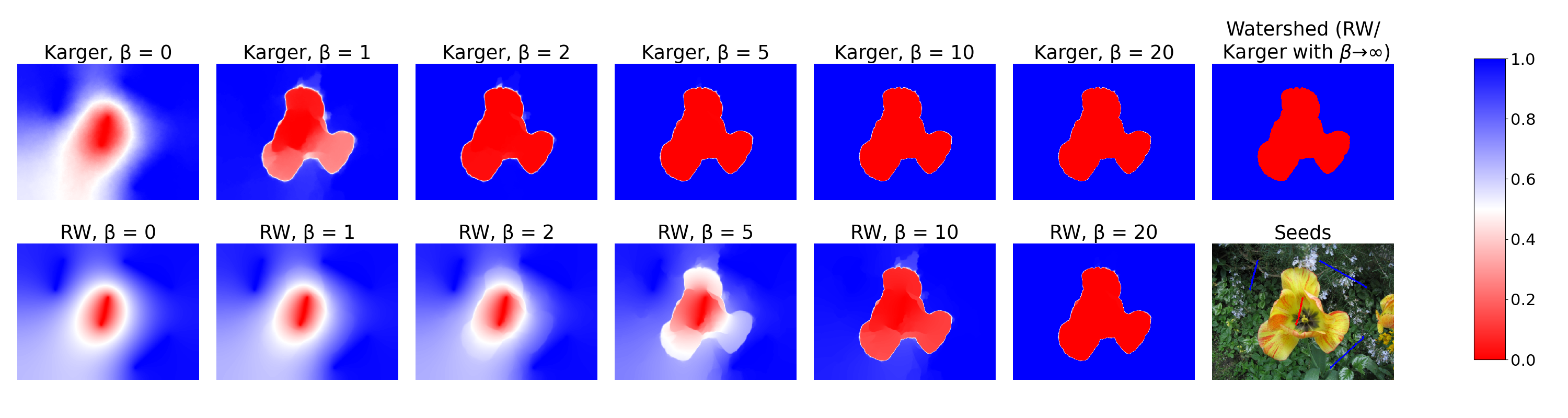}
  \caption{The probabilistic estimates of the Karger potential, compared to those
  of the random walker / harmonic energy minimization, for various graph edge weights.}
  \label{fig:potential}
\end{figure*}

In addition to the random walker, we also compare to the watershed segmentation,
which has been used for both seeded segmentation~\cite{cousty_2009} and
semi-supervised learning~\cite{challa2019}. This segmentation arises from a
maximum spanning forest that separates the seeds~\cite{cousty_2009}. If there is
only one maximum spanning forest, both the Karger potential and the random
walker potential converge to this segmentation as \(\beta \to \infty\). The more general
case of multiple maximum spanning forests is described by the Power Watershed
framework~\cite{coupriePowerWatershedUnifying2011,najman2017}, which generalizes
both the watershed and the random walker. This framework has two parameters,
\(q\) and \(p\), and the case \(q = 2, p \to \infty\) is the limit of the random walker
for \(\beta \to \infty\), without any assumptions on the number of maximum spanning
forests. When there \emph{is} a unique maximum spanning forest, Power Watershed reduces
to watershed; in particular, this is the case if all the edge weights are distinct.
So we rounded the edge weights to 8 bits to artificially introduce the
edges with equal weight that give Power Watershed the opportunity to shine
relative to watershed. This leads to 256 different possible edge weights,
exactly as in~\cite{coupriePowerWatershedUnifying2011}.

\begin{figure*}[htb]
  \centering
  \includegraphics[width=\linewidth]{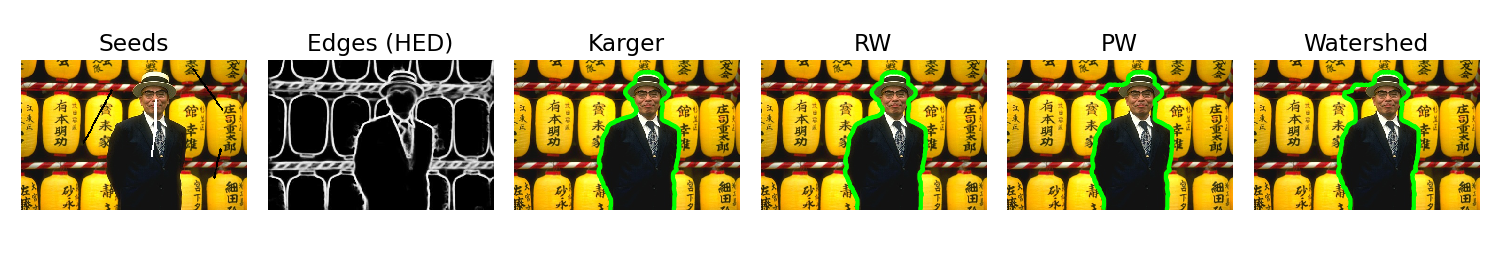}
  \includegraphics[width=\linewidth]{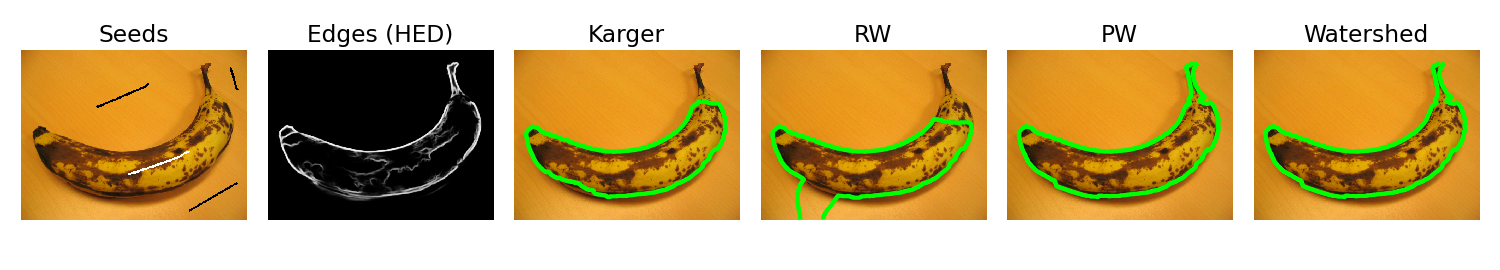}
  \caption{Some qualitative differences between Karger-type contractions, random walker / harmonic energy
    minimization and (Power) watershed. \(\beta\) values are chosen optimally
    for each method.}\label{fig:segmentations}
\end{figure*}

\begin{table}
  \centering
  \begin{tabular}{l l l l}
    \toprule
    & \(\uparrow\) ARI
    & \(\uparrow\) Acc [\%]
    & \(\downarrow\) VoI \\
      \midrule
    Contraction & \(0.82 \pm 0.02\) & \(96.3 \pm 0.6\) & \(0.24 \pm 0.02\) \\
    RW & \(0.82 \pm 0.02\) & \(96.0 \pm 0.6\) & \(0.25 \pm 0.03\) \\
    Watershed & \(0.83 \pm 0.02\) & \(96.2 \pm 0.5\) & \(0.24 \pm 0.02\) \\
    Power WS & \(0.83 \pm 0.02\) & \(96.2 \pm 0.5\) & \(0.24 \pm 0.02\)\\
    \bottomrule
  \end{tabular}
  \caption{Seeded segmentation: Mean adjusted Rand index (ARI), accuracy (Acc)
    and variation of information (VoI) on the Grabcut dataset. RW = Random
    Walker, Power WS = Power Watershed with \(q = 2, p \to \infty\)}
  \label{tab:grabcut}
\end{table}

Table~\ref{tab:grabcut} shows the results on the entire Grabcut dataset. We optimized
\(\beta\) by hand separately for each method and used the optimal values \(\beta = 10\) for
the Karger potential and \(\beta = 20\) for the random walker. However, the performance of both
algorithms is relatively stable within this range of values. The watershed algorithm does not
depend on the value of \(\beta\), as long as \(\beta > 0\). For Power Watershed, we used
\(\beta = 10\) (though its dependency on \(\beta\) is very low anyway).
The reported error is the standard error
of the mean over the dataset. We used 1000 runs of the seeded segmentation algorithm to
approximate the Karger potential, which made the approximation error negligible in comparison.

We compare the four methods using the Adjusted Rand Index (ARI), their classification accuracy and the Variation of Information (VoI).
For ARI and accuracy, higher is better, for VoI, lower is better. All metrics are calculated only over the
unlabeled pixels.

Figure~\ref{fig:segmentations} shows an example of the seeds that were used, the output of the
edge detection network and the resulting segmentations for each of the four methods, each at their optimal \(\beta\) values.
The results are for the most part very similar -- the segmentations shown here have been selected because
they are visibly different. In the first row, there are many strong edges and
the (Power) watershed follows a different edge
than the other methods. In the second row, some edges are missing and the four methods respond differently
to this \enquote{leak}.

The output of the contraction algorithm is only an approximation of the true Karger potential but the error is so small that it
does not visibly affect the contours of the segmentation.

\paragraph{Semi-supervised learning}
Here, we used classical benchmark data from the training set of the USPS handwritten digits
dataset~\cite{hull1994,lecun1990}.
These are labeled \(16 \times 16\) grayscale images of digits from 0 to 9. We calculated all pairwise euclidean distances between the
images and built the 10-nearest neighbors graph based on those. The graph weights were again computed using a radial basis
function,
\begin{equation}
  w_{ij} = \exp\left(-\beta \frac{d_{ij}^2}{a^2}\right)\,,
\end{equation}
with \(a := \max_{\{i, j\} \in E} d_{ij}\), where \(d_{ij}\) are the euclidean distances.

We used random subsets of different sizes as labeled vertices and left the remaining vertices to be labeled.
For each size of the labeled set, we sampled 20 sets. Table~\ref{tab:usps} shows the accuracies over
unlabeled data, averaged over these 20 samples. The errors are the standard errors of the sample mean.
As before, the \(\beta\) values were chosen individually
for each method to maximize performance (\(\beta = 5\) for the random walker,
\(\beta = 2\) for the contraction method).
\begin{table}
  \setlength\tabcolsep{4pt}
  \centering
  \scalebox{0.9}{
  \begin{tabular}{l r r r r}
    \toprule
    Seeds & 20 & 40 & 100 & 200 \\
    \midrule
    Contraction & \(\mathbf{62.2} \pm 1.8\) & \(\mathbf{73.1} \pm 1.5\) & \(\mathbf{89.0} \pm 0.5\) & \(\mathbf{92.6} \pm 0.2\)\\
    RW & \(53.7 \pm 1.9\) & \(68.0 \pm 1.2\) & \(87.7 \pm 0.7\) & \(\mathbf{92.4} \pm 0.3\) \\
    Watershed & \(54.3 \pm 2.0\) & \(58.7 \pm 1.9\) & \(74.3 \pm 1.0\) & \(80.0 \pm 0.9\) \\
    Power WS & \(54.4 \pm 2.0\) & \(59.0 \pm 1.9\) & \(74.7 \pm 1.0\) & \(80.8 \pm 0.9\) \\
    \bottomrule
   \end{tabular}}
  \caption{Graph-based semi-supervised learning: Accuracies in \% on the USPS
    dataset. RW = Random Walker, Power WS = Power Watershed with \(q = 2, p \to \infty\)}
  \label{tab:usps}
\end{table}

Throughout, we used the scikit-image implementation of the random walker~\cite{waltScikitimageImageProcessing2014}
with slight adaptations to use the edge weights described above.

\paragraph{Results}
In all our experiments, the new method based on the Karger potential performed
comparably to the random walker / harmonic energy minimization. However, the new method
has significantly better results in the semi-supervised learning setting with few labeled vertices.

\section{Conclusion}
We have shown that contraction algorithms that are continuous or that use only
local properties of the edges cannot efficiently solve the \(s\)-\(t\)-mincut
problem or the normalized cut problem to optimality. On the other hand, we have
demonstrated that certain extensions of Karger's algorithm \emph{can} be
successfully used for seeded segmentation and semi-supervised learning tasks: we
have presented a contraction-based algorithm that performs as well as or better
than the random walker / harmonic energy minimization, while having an
asymptotic time complexity linear in the number of edges.

Future work might address the question whether contraction algorithms based on global
properties can be useful for solving the \(s\)-\(t\)-mincut problem or whether our result
can be extended to an even wider class of algorithms. Another open question is whether
the \(s\)-\(t\)-contraction algorithm can find \(s\)-\(t\)-mincuts quickly on graphs that occur
in practice, as opposed to the \enquote{malicious} artificial graphs we used in the impossibility proofs.
Finally, Karger's algorithm induces a distribution over spanning 2-forests, similarly to
the distribution we describe in Section~\ref{sec:seeded_segmentation}. Future research
could shed more light on this distribution, for example whether it is uniquely well suited
for finding minimum cuts or whether a Gibbs distribution would yield a result similar
to Theorem~\ref{thm:karger_probability_bound}.

\section*{Acknowledgements}
This work is funded by the Deutsche Forschungsgemeinschaft (DFG, German Research
Foundation) under Germany's Excellence Strategy EXC 2181/1 - 390900948 (the
Heidelberg STRUCTURES Excellence Cluster).

{\small
\bibliographystyle{ieee_fullname}
\bibliography{references}
}

\onecolumn
\begin{changemargin}{2cm}{2cm}
\begin{appendices}
\section*{Appendix}
\section{Proofs for continuous contraction algorithms}
\label{app:continuous}
We will first prove theorem~\ref{thm:continuous_impossibility} for \(s\)-\(t\)-mincuts. Afterwards,
we will give those parts of the proof for the normalized cut version that differ from the
version for \(s\)-\(t\)-mincuts.

\begin{proof}[Proof of theorem~\ref{thm:continuous_impossibility} for \(s\)-\(t\)-mincuts]
  Fix a value of \(n\) and a continuous contraction algorithm, we will then construct a graph on which
  this algorithm finds an \(s\)-\(t\)-mincut with probability \(\leq 2^{-n + 3}\).

  Consider the graph \(G\) in fig.~\ref{fig:simple_counterexample} in the main paper, but with weight 1 for all edges
  and generalized to \(n\) vertices (two of which are \(s\) and \(t\)).
  We call the adjacency matrix of this graph \(A\). Every \(s\)-\(t\)-cut of \(G\) is minimal, and
  there are \(2^{n - 2}\) such \(s\)-\(t\)-cuts.
  So there must be at least one \(s\)-\(t\)-mincut, whose cut set we shall call \(C\),
  that the contraction algorithm selects with probability \(\leq 2^{-n + 2}\).

  The idea of this proof is to slightly decrease the weights of all the edges in \(C\). Because
  of continuity, we can do this in such a way that the probability of selecting \(C\) does not
  change by much. Then the cut defined by \(C\) will be the unique \(s\)-\(t\)-mincut
  in the modified graph but will still be selected with probability of order \(2^{-n}\).
  What follows is a more rigorous version of this argument.

  We will write \(p_C(\tilde{A})\) for the probability that the algorithm selects the cut \(C\)
  on the graph with \(n\) vertices and weighted adjacency matrix \(\tilde{A}\). We know
  that \(p_C(A) \leq 2^{-n + 2}\). Our goal is to find an adjacency matrix \(A'\) for which \(C\)
  is the unique \(s\)-\(t\)-mincut and \(p_C(A') \leq 2^{-n + 3}\).

  First, we need to show that for a continuous contraction algorithm,
  \(p_C(\tilde{A})\) is a continuous function of \(\tilde{A}\).
  The definition of continuous contraction algorithms only states that the scores
  at each step need to be continuous function of the adjacency matrix.
  It's unsurprising that this also leads to continuous overall probabilities
  of selecting given cuts. The details do not provide much insight and are
  shown separately as lemma~\ref{thm:continuity_contraction}.

  This continuity of \(p_C(\tilde{A})\) means that for \(\eps := 2^{-n + 2}\),
  there is a \(\delta > 0\) such that if \(\norm{A - A'}_\infty \leq \delta\) for some \(A'\), then
  \(\abs{p_C(A) - p_C(A')} \leq \eps\).

  So we define the graph \(G'\) with adjacency matrix \(A'\) by setting the weights of the edges in the cut set
  \(C\) to \(1 - \delta\) and leaving the other weights at \(1\)\footnote{We can of course assume
    \(\delta < 1\) without loss of generality}. Then \(\norm{A - A'}_\infty = \delta\),
  so
  \[\abs{p_C(A) - p_C(A')} \leq \eps = 2^{-n + 2}\]
  This means that the probability of finding the cut \(C\) in the new graph \(G'\) is
  \[p_C(A') \leq p_C(A) + \eps \leq 2^{-n + 2} + 2^{-n + 2} = 2^{-n + 3}\]

  At the same time, \(C\) is the \emph{unique} \(s\)-\(t\)-mincut of \(G'\).
  Every \(s\)-\(t\)-cut has a cut set with the same cardinality and \(C\) is the only
  one which contains only edges that have weight \(1 - \delta\) --- every other cut set
  contains some edges with weight 1.

  This means that on \(G'\), the contraction algorithm has only an exponentially low
  probability of finding \emph{any} \(s\)-\(t\)-mincut, as claimed.
\end{proof}

\begin{proof}[Proof of theorem~\ref{thm:continuous_impossibility} for normalized cuts]
  Instead of the graph from fig.~\ref{fig:simple_counterexample} in the main paper that we used
  in the previous proof,
  let \(G\) be a complete unweighted graph on \(n\) vertices. We will show that in this graph, every cut
  is a normalized cut.

  In a complete unweighted graph, we always have
  \[w(A, B) = \sum_{a\in A, b\in B} w_{ab} = |A| \cdot |B| - |A \cap B|\]
  (the last term is necessary because there are no self-loops).
  Therefore, with \(k := |A|\), we get
  \begin{align*}
    w(A, A^c) &= k(n - k) \\
    w(A, V) &= k(n - 1) \\
    w(A^c, V) &= (n - k)(n - 1)
    \end{align*}
  So the normalized cut cost of the cut \((A, A^c)\) is
  \begin{equation*}
    \begin{split}
      \operatorname{ncut}(A, A^c) &= \frac{w(A, A^c)}{w(A, V)} + \frac{w(A, A^c)}{w(A^c, V)}
      = \frac{k(n - k)}{k(n - 1)} + \frac{k(n - k)}{(n - k)(n - 1)}\\
      &= \frac{(n - k) + k}{n - 1} = \frac{n}{n - 1}
    \end{split}
  \end{equation*}
  for each of the \(2^{n - 1}\) possible cuts. This means that there are \(2^{n - 1}\) normalized cuts,
  and the algorithm must assign probability \(\leq 2^{-n + 1}\) to at least one of them.

  From here on the proof procedes like that for \(s\)-\(t\)-mincuts:
  if the weights are slightly perturbed, there will be a unique normalized cut,
  but its probability will still be close to \(2^{-n + 1}\). We therefore don't
  repeat the details.
\end{proof}

To make working with weighted adjacency matrices easier, we consider
every graph to be fully connected for the following Lemma. Non-existent edges
are instead treated as edges with weight zero.
\begin{lemma}\label{thm:continuity_contraction}
  Let \(C\) be a fixed set of edges of the complete graph on \(n\) vertices. Let \(p_C(A)\) be the probability that a given continuous
  contraction algorithm does not contract any edges from \(C\) when run on the graph with \(n\) vertices
  and weighted adjacency matrix \(A\). Then \(p_C\) is a continuous function.
\end{lemma}
Note that we don't require \(C\) to be a cut set and that here, \(p_C(A)\) does not always denote the probability
that \(C\) is the chosen cut. This makes the proof more concise and is a strict generalization: if \(C\) happens
to be a cut set for some adjacency matrix \(A\), then \(p_C(A)\) will be the probability that \(C\) is the
final chosen cut.
\begin{proof}
  Let \(e_1, \ldots, e_{n - 2}\) be an arbitrary but fixed set of edges. We will show that the probability \(p(e_1, \ldots, e_{n - 2})\)
  that the algorithm contracts \(e_1, \ldots, e_{n - 2}\) in that order
  is a continuous function of \(A\). Then the claim follows because \(p_C(A)\) is simply the sum of
  these probabilities over all  \((n - 2)\)-tuples of edges that don't contain edges from \(C\).
  
  We prove the claim in two steps:
  \begin{enumerate}
  \item
    Let \(C_k(A)\) for \(k \in \{1, \ldots, n - 2\}\) be the adjacency matrix that is reached from starting
    with \(A\) and contracting \(e_1, \ldots, e_k\). We will show that \(C_k(A)\) is a continuous function
    for all \(k\).
  \item We then show that \(p(e_1, \ldots, e_{n - 2})\) is a continuous function of the partially contracted
    adjacency matrices \(C_1(A), \ldots, C_{n - 2}(A)\)
  \end{enumerate}
  It will then follow that \(p(e_1, \ldots, e_{n - 2})\) and thus \(p_C\) is a continuous function of \(A\), as a composition
  of continuous functions.

  For the first step, note that contracting an edge sets an entry in the adjacency matrix to zero
  and adds its previous value to another entry. This means that each entry of \(C_1(A)\)
  is either zero (for all \(A\)) or a certain sum of entries of \(A\). The structure of the sum is determined by \(e_1\)
  and does not depend on \(A\). Therefore, \(C_1\) is a continuous function of \(A\).
  Each entry in \(C_2(A)\) is zero or a sum of entries in \(C_1(A)\) and therefore a continuous function of \(C_1(A)\),
  which makes it a continuous function of \(A\) by composition. By induction it follows that all \(C_k\)
  are continuous functions of \(A\).

  For the second step, we write \(p(e_1, \ldots, e_{n - 2})\) as
  \begin{equation}
    p(e_1, \ldots, e_{n - 2}) = \prod_{k = 1}^{n - 2} \frac{\mathcal{W}(e_k; C_k(A))}{\sum_{e}\mathcal{W}(e; C_k(A))}
  \end{equation}
  the sum over \(e\) is over all the edges at that step; which summands appear depends only on \(e_1, \ldots, e_{n - 2}\)
  and not on \(A\).
  The scores \(\mathcal{W}\) that appear are by definition continuous in their second argument \(C_k(A)\).
  Since \(C_k(A)\) is continuous in \(A\) and the entire expression is clearly continuous in the scores,
  \(p(e_1, \ldots, e_{n - 2})\) is continuous in \(A\) as claimed.
\end{proof}

\section{Proofs for local contraction algorithms}
\label{app:local}
We will first prove theorem~\ref{thm:local_impossibility} for \(s\)-\(t\)-mincuts. The approximability
result from corollary~\ref{thm:local_approximability} and the result for normalized cuts will then follow
easily.

\begin{proof}[Proof of theorem~\ref{thm:local_impossibility} for \(s\)-\(t\)-mincuts]
The graphs \(G_n\) we will use consist of \(n\) copies of each of three different
subgraphs, shown in fig.~\ref{fig:bands}. As an example, a schematic version
of \(G_2\) is shown in fig.~\ref{fig:g_2}.
Each of the colored boxes contains one subgraph, each of the three types occurs twice. The different box colors
denote the three different types. The last two types
only differ in their orientation but we will treat them separately.
The general graph \(G_n\) simply has \(n\) instead of two copies of each subgraph.

\begin{figure}
  \begin{subfigure}{\linewidth}
    \centering
    \includegraphics[width=.7\linewidth]{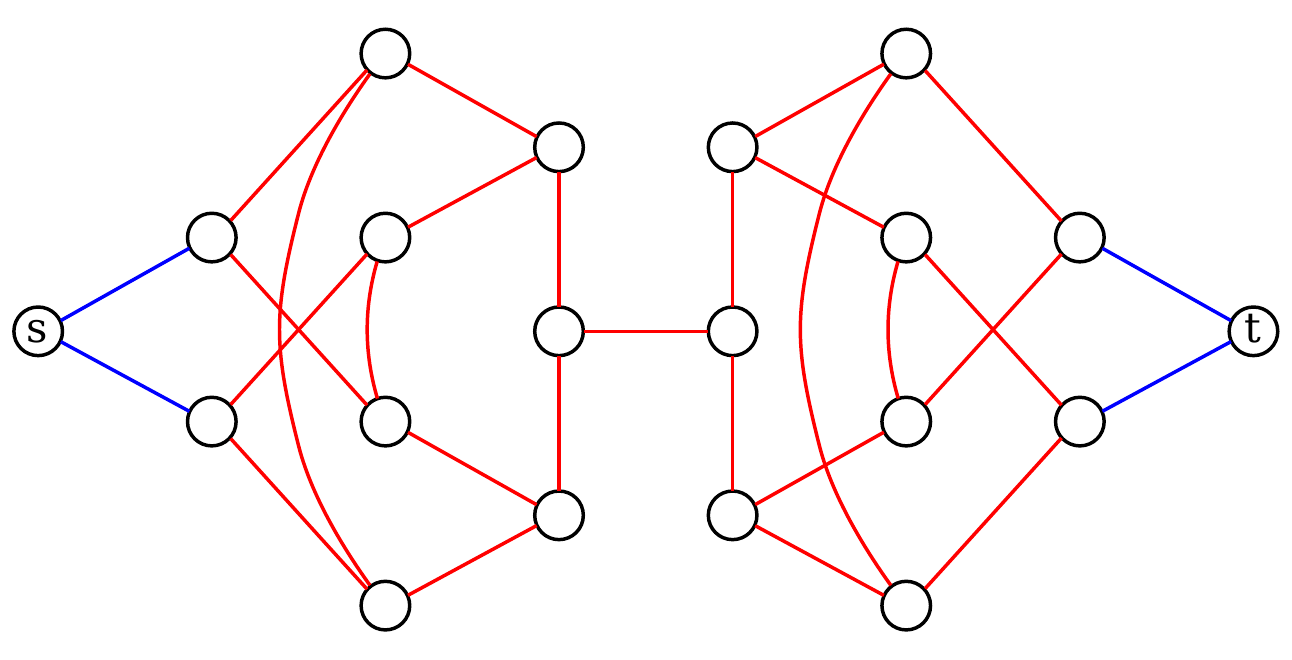}
    \caption{Band of type A}
  \end{subfigure}
  \begin{subfigure}{\linewidth}
    \centering
    \includegraphics[width=.7\linewidth]{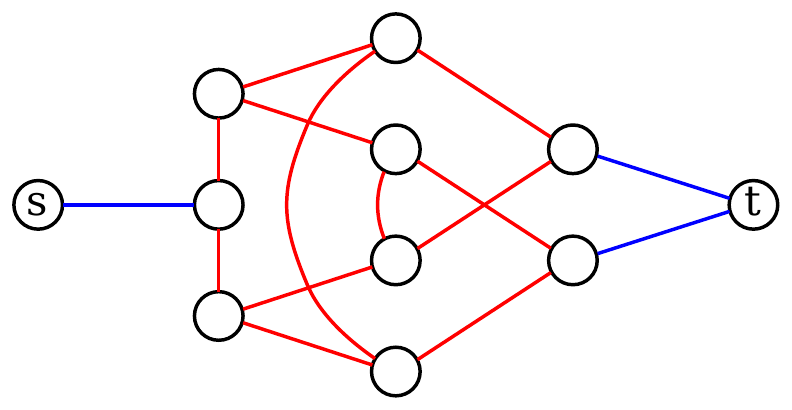}
    \caption{Band of type B}
  \end{subfigure}
  \begin{subfigure}{\linewidth}
    \centering
    \includegraphics[width=.7\linewidth]{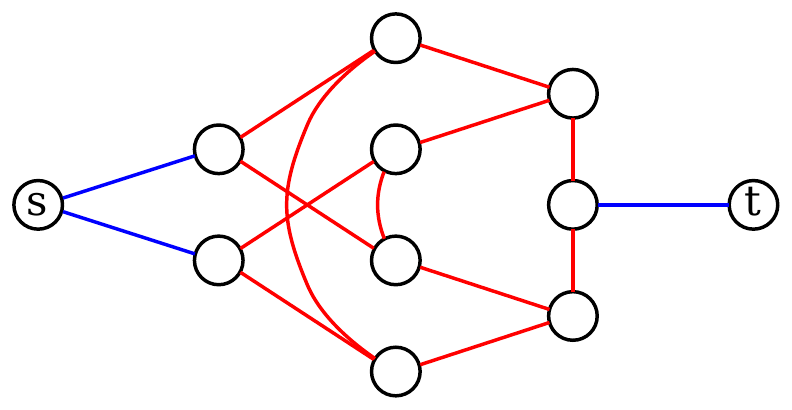}
    \caption{Band of type C}
  \end{subfigure}
  \caption{The three different \enquote{bands} used in \(G_n\)}
  \label{fig:bands}
\end{figure}

\begin{figure}[htbp]
\centering
\includegraphics[height=.85\textheight]{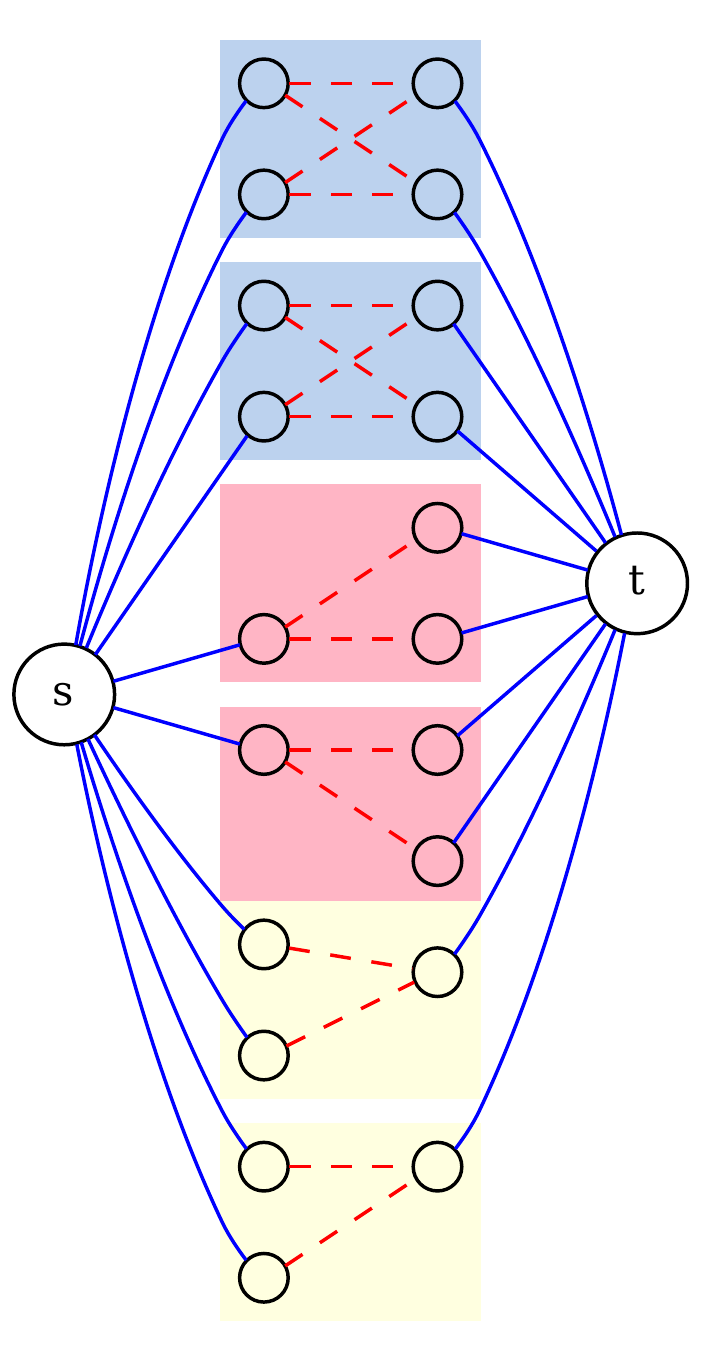}
\caption{\label{fig:g_2}The unweighted graph \(G_2\). The idea is still the same as for the simpler example from
  fig.~\ref{fig:simple_counterexample} in the main paper. But the short parallel paths between \(s\) and \(t\) have
  now been replaced by \enquote{bands}, our name for the subgraphs in colored boxes. This construction
  ensures that it is impossible to find \enquote{safe} edges for contraction based only on
  local properties. The bands are only hinted at here, the full bands are shown in fig.~\ref{fig:bands}.
  Blue corresponds to type A, red to B, yellow to C.}
\end{figure}

We call each such subgraph, including the blue edges that connect it to \(s\) and
\(t\), a \emph{band}. There are \(3n\) bands, \(n\) of each type.
We say that a band has been \emph{touched} if one of the edges belonging to it has been contracted.
Clearly, every band has to be touched at some point during the contraction process.

One of the key ideas of the proof is that \(G_n\) contains many edges that have isomorphic
neighborhoods, but some of which are part of the \(s\)-\(t\)-mincut while others are not.
A local contraction algorithm assigns the same score to each of these edges with isomorphic
neighborhoods. This will give us a lower bound on the contraction probability of edges
included in the \(s\)-\(t\)-mincut cut set in any particular step.

All of the edges in \(G_n\) belong to one of three isomorphism classes of neighborhoods,
which we call \emph{red}, \emph{\(s\)-blue} and \emph{\(t\)-blue}. These neighborhoods
are shown in fig.~\ref{fig:neighborhoods}.

\begin{figure}
  \centering
  \begin{subfigure}{.45\linewidth}
    \centering
    \includegraphics[width=\linewidth]{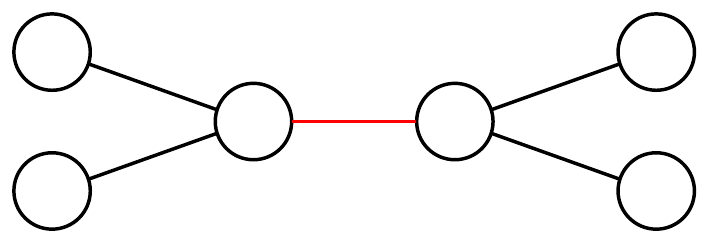}
    \caption{Red neighborhood}\label{fig:red_neighborhood}
  \end{subfigure}
  \begin{subfigure}{.45\linewidth}
    \centering
    \includegraphics[width=\linewidth]{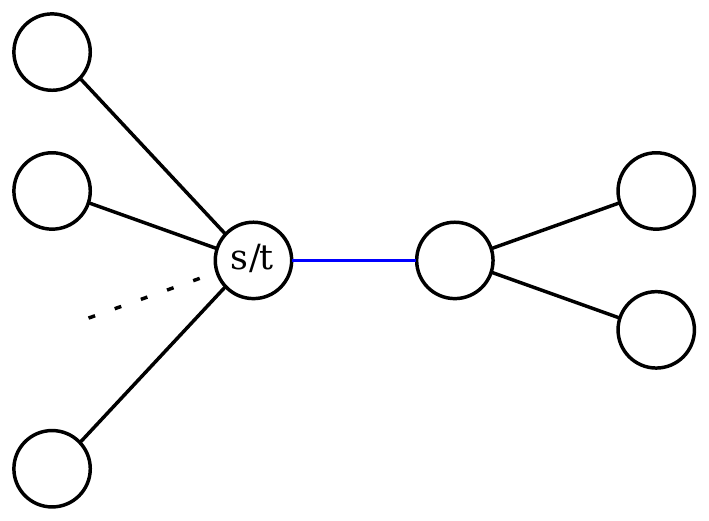}
    \caption{Blue neighborhood}\label{fig:blue_neighborhood}
  \end{subfigure}
  \caption{The different neighborhoods that occur in \(G_n\). \(s\)-blue and
    \(t\)-blue are both shown in fig.~\ref{fig:blue_neighborhood} because they differ
  only in whether they contain the node \(s\) or \(t\).}\label{fig:neighborhoods}
\end{figure}

The blue neighborhood classes depend on the degree of \(s\) (\(t\)) which
is a function of \(n\). We will call an edge \(s\)-blue (\(t\)-blue) if its neighborhood
fits the schema from fig.~\ref{fig:blue_neighborhood}, no matter the degree
of \(s\) (\(t\)). In a fixed graph, all \(s\)-blue (\(t\)-blue) edges have isomorphic
neighborhoods. That the same is not true across graphs does not matter for our
purposes.

All the edges are colored according to their neighborhood (blue and red)
in fig.~\ref{fig:bands}.
During the contraction process, other neighborhoods may of course arise.

The unique \(s\)-\(t\)-mincut of \(G_n\) cuts the red edge in the middle of each band
of type A and the blue edges on the side where there is only one of them in bands of
type B and C. The contraction algorithm will find this minimum cut iff it does not contract
any of these edges. So we will say a contraction is \enquote{wrong} or a \enquote{mistake} if it
contracts one of these edges belonging to the \(s\)-\(t\)-mincut.

We will now prove some useful statements:
\begin{enumerate}
\item In an untouched band, all edges are either red, \(s\)-blue or \(t\)-blue, with each
edge belonging to the same type as in the original graph \(G_n\).
\begin{subproof}
It's clear that contractions of red edges in one band don't influence other bands. If a blue
edge is contracted, this can change the degree of \(s\) or \(t\) but has no
influence on other bands apart from that.
\end{subproof}
\item If at most \(\frac{n}{2}\) bands have been touched, then the probability that
contracting a blue or red edge is wrong is
\(p(\text{wrong}|\text{red or blue}) \geq \frac{1}{118}\).
\begin{subproof}
There are still at least \(\frac{n}{2}\) untouched bands of all three types. Each type contains
a red edge, an \(s\)-blue edge or a \(t\)-blue edge that mustn't be contracted respectively
(as per the statement proven just above).
So there are at least \(\frac{n}{2}\) wrong red edges, \(\frac{n}{2}\) wrong \(s\)-blue
edges and \(\frac{n}{2}\) wrong \(t\)-blue edges.

Because all red edges have the same neighborhood, the local contraction algorithm
assigns the same score to all red edges. The same is true for \(s\)-blue edges and
for \(t\)-blue edges. So we have
\[p(\text{wrong}|\text{red}) = \frac{\#\text{wrong red edges}}{\#\text{red edges}}
\geq \frac{\#\text{wrong red edges}}{\#\text{total edges}} \geq \frac{n/2}{59n} = \frac{1}{118}\]
and similarly \(p(\text{wrong}|s\text{-blue}) \geq \frac{1}{118}\) and the same
for \(t\). This means that \(p(\text{wrong}|\text{red or blue})\) is bounded by
\begin{equation*}
\begin{split}
p(\text{wrong}|\text{r or b})
&= \frac{p(\text{wrong}|\text{r})p(\text{r}) + p(\text{wrong}|s\text{-b})p(s\text{-b}) 
+ p(\text{wrong}|t\text{-b})p(t\text{-b})}{p(\text{r}) + p(s\text{-b}) + p(t\text{-b})}\\
&\geq \frac{\frac{1}{118}\cdot p(\text{r}) + \frac{1}{118}\cdot p(s\text{-b}) 
+ \frac{1}{118}\cdot p(t\text{-b})}{p(\text{r}) + p(s\text{-b}) + p(t\text{-b})}\\
&= \frac{1}{118}
\end{split}
\end{equation*}
\(p(r)\), \(p(s\text{-b})\) and \(p(t\text{-b})\) are what can be influenced by the
choice of the scoring function \(\mathcal{W}\) but these terms cancel as we can see.
\end{subproof}
\end{enumerate}

We will now prove inductively that contracting an edge in \(k\) different bands
without contracting any wrong edges happens with probability \(\leq \left(\frac{117}{118}\right)^k\)
for \(k \leq \frac{n}{2}\):
\[p_k := p\left(\text{no mistakes} \mid k \text{ bands touched}\right) \leq \left(\frac{117}{118}\right)^k\]

\begin{subproof}
\begin{description}
\item[\(k = 0\)] Nothing to show.

\item[\(k \to k + 1\)] The probability of not making any mistakes until \(k + 1\) bands
  have been touched is the probability \(p_k\) of correctly touching the first \(k\)
  bands times the probability of not making a mistake while touching the final band.

  From statement 1 proven above, we know that to touch a new band, a red, \(s\)-blue
  or \(t\)-blue edge will have to be contracted at some point. From statement 2 we know
  that the probability of making a mistake on that single contraction is at least \(\frac{1}{118}\).
  Additional contractions may be made, but they cannot decrease the total probability of making any mistake.
  So
  \[p_{k + 1} \leq \left(\frac{117}{118}\right)^k \cdot \frac{117}{118} = \left(\frac{117}{118}\right)^{k + 1}\]
  which proves the claim for \(k + 1\).
\end{description}
\end{subproof}

Since all bands have to be touched eventually, we can apply this statement with
\(k = \frac{n}{2}\).
So the success probability is at most \(\left(\frac{117}{118}\right)^{n/2}\) which is
exponentially low in the number of vertices, \(36n + 2\), as claimed.
\end{proof}

\begin{proof}[Proof of corollary~\ref{thm:local_approximability}]
  Since the \(s\)-\(t\)-mincut has cost \(3n\), an \(\alpha\)-minimal cut may be worse
  than the mincut by at most \(3n(\alpha - 1)\).
  Every wrong contraction in an untouched band increases the cost of the best cut that is still
  possible by at least 1 (because there is only one unique way to optimally
  cut each band). So to find an \(\alpha\)-minimal \(s\)-\(t\)-cut, at most \(3n(\alpha - 1)\)
  wrong contractions may be made in untouched bands.

  As there are \(3n\) bands, at least \(3n - 3n(\alpha - 1) = (6 - 3\alpha)n\) contractions
  in different bands must be made without mistakes.

  We showed in the proof of theorem~\ref{thm:local_impossibility} that making contractions in
  \(k\) different bands without mistakes (with \(k \leq \frac{n}{2}\)) happens with probability
  \(\leq \left(\frac{117}{118}\right)^k\). If \(\alpha < 2\), then \(6 - 3\alpha > 0\),
  and therefore we can apply this result\footnote{If \((6 - 3\alpha)n > \frac{n}{2}\), we just use \(k = \frac{n}{2}\)}
  with \(k = (6 - 3\alpha)n\) and see that the
  probability of correctly contracting edges in the required number of bands is exponentially
  low in \(n\).

  Therefore, the probability that mistakes are made in only \(3n(\alpha - 1)\) bands is
  exponentially low, and thus also the probability of finding an \(\alpha\)-minimal \(s\)-\(t\)-cut.
\end{proof}

\begin{proof}[Proof of theorem~\ref{thm:local_impossibility} for normalized cuts]
It suffices to show that the \(s\)-\(t\)-mincut in \(G_n\) is also the normalized cut
for large \(n\). The \(s\)-\(t\)-mincut cuts \(3n\) edges. Because the partitions are
perfectly balanced in terms of internal edge weights, only cuts that cut fewer edges
than that can have a lower normalized cut cost. In particular, any such cut could not
separate \(s\) and \(t\). One of its partitions could therefore be no larger than
one of the bands between \(s\) and \(t\). But the normalized cut cost of such a cut
approaches 1 for large \(n\), whereas the ncut cost of the \(s\)-\(t\)-mincut is always
\(2 \cdot \frac{3n}{2 \cdot 28n + 3n} = \frac{6n}{59n} = \frac{6}{59}\)\footnote{
  Each partition has \(28n\) internal edges and there are \(3n\) edges in the cut between partitions}.
So the \(s\)-\(t\)-mincut is indeed also the normalized cut for large \(n\).
\end{proof}
\section{Implementation of the seeded contraction algorithm}
For unweighted graphs, Karger's algorithm can be implemented in \(\mathcal{O}(m)\) time
as follows~\cite{kargerNewApproachMinimum1996}:
First, a random permutation of all \(m\) edges is generated which takes \(\mathcal{O}(m)\) time.
Afterwards, edges are contracted in the chosen order until only two
vertices remain. If an edge has already been removed by previous contractions, it is skipped.
\cite{kargerNewApproachMinimum1996} also describes how this method can be generalized to weighted graphs. The only
change is in how to generate the permutation of edges to give different probabilities to
different permutations.

To keep track of when to stop and of the current segmentation at each step, we use
a union-find data structure. Keeping this structure up to date increases the runtime
to \(\alpha(n)\mathcal{O}(m)\)~\cite{tarjanWorstcaseAnalysisSet1984}
where \(n\) is the number of vertices and \(\alpha(n)\) the inverse Ackermann function.
But since \(\alpha(n) < 5\) for all practical values of \(n\),
this theoretical increase has no practical relevance.

Two modifications are necessary to adapt this implementation to the seeded contraction
algorithm: first, we initialize the union-find data structure such that nodes with the same
seed are in the same cluster from the beginning. This is possible with a linear scan over
all nodes in \(\mathcal{O}(n)\).

Second, we keep a boolean vertex property updated
that denotes whether a node is already labeled (i.e.\ in the cluster of a seed node) or not.
Whenever we come to an edge connecting two nodes that are already labelled, we skip it
instead of merging these nodes. This ensures that no seeds with different labels end up
in the same cluster and each node has a well-defined label at the end. These extensions
do not increase the runtime of processing one edge beyond \(\mathcal{O}(1)\), so the
total runtime of the algorithm stays \(\mathcal{O}(m\alpha(n))\).

\section{Details on experiments}
\subsection{Metrics}
We used three common metrics to evaluate performance in the Grabcut experiment:
the adjusted Rand index (ARI), variation of information (VoI) and accuracy.

The (unadjusted) Rand index is defined as the accuracy on the space of pairs of
samples, in the following sense: count the number TP of pairs of samples that are
correctly put in the same cluster (\emph{true positives}) and the number TN of pairs
of samples that are correctly put in different clusters (\emph{true negatives}).
Then the Rand index is \(\frac{TP + TN}{{n \choose 2}}\) for \(n\) samples,
where \({n \choose 2}\) is the total number of pairs of samples.

The adjusted Rand index renormalizes the Rand index such that it is 1 for a
perfect clustering and has expected value 0 for a random clustering, independently of the number
of clusters. This is done by correcting for chance with
\[\text{ARI } = \frac{\text{RI } - \mathbb{E}[\text{RI}]}{1 - \mathbb{E}[\text{RI}]}\]
where RI is the Rand index and the expectation is over permutations of the assigned labels.
The expected value of 0 makes the ARI easy to interpret compared to the unadjusted Rand index,
for which even random clusterings can have a positive expected value.

The variation of information between two variables \(X\) and \(Y\) can be defined as
\[\operatorname{VI}(X; Y) := H(X|Y) + H(Y|X)\]
where \(H(X|Y)\) is the conditional Shannon entropy. To get a clustering metric, we
let \(X\) and \(Y\) be the different labelings (in our case ground truth and the labeling to be evaluated).
The joint distribution over \(X\) and \(Y\) is defined by picking samples uniformly at
random.

\subsection{Parameters}
In both experiments, we optimized the \(\beta\) parameter by hand to maximize the performance
according to the metrics we used separately for each method. In the Grabcut experiment multiple
metrics were used, but they all reached their maximum at the same \(\beta\) value of those
we tested.

In both experiments, we first found a reasonable range of \(\beta\) values and then tested ten different
values within these ranges. For the Grabcut experiment, this range was \(\beta = 5\) to \(\beta = 100\),
for the USPS experiment it was \(\beta = 1\) to \(\beta = 20\).

\subsection{Empirical runtimes}
For a complete run on both datasets (Grabcut and USPS), the new Karger-based
algorithm takes about 20 minutes\footnote{with 100 runs, enough for a reasonably
good approximation of the true potential}, the random walker takes about 5
minutes and watershed 9 minutes.
All of these are wall clock times when running on 6 CPU cores. The random walker
implementation is the one used in SciPy, the other algorithms were implemented
by us in Julia. Their runtimes could probably be decreased with a more efficient
implementation.

\end{appendices}
\end{changemargin}
\end{document}